\documentclass[pdflatex,sn-mathphys-num]{sn-jnl}% Math and Physical Sciences Numbered Reference Style 
\usepackage{graphicx}%
\usepackage{multirow}%
\usepackage{amsmath,amssymb,amsfonts}%
\usepackage{amsthm}
\usepackage{mathrsfs}%
\usepackage[title]{appendix}%
\usepackage{xcolor}%
\usepackage{textcomp}%
\usepackage{manyfoot}%
\usepackage{booktabs}%
\usepackage{algorithm}%
\usepackage{algorithmicx}%
\usepackage{algpseudocode}%
\usepackage{listings}%
\usepackage{subcaption}
\usepackage{placeins}
\usepackage{float}
\usepackage{graphicx}

\theoremstyle{plain} % Default style for theorems, lemmas, corollaries, propositions
\newtheorem{theorem}{Theorem} % meant for continuous numbers
\theoremstyle{definition} % Style for definitions
\newtheorem{definition}{Definition}

\theoremstyle{remark} % Style for remarks and examples

\begin{document}

\title[Article Title]{Communications Performance Analysis of Wireless Multiple Access Channel with Specially Correlated Sources}


% Author information
\author [1]{\fnm{Akram} \sur{Entezami}}\email{Akram\_Entezami1@baylor.edu}

\author[2]{\fnm{Ghoosheh} \sur{Abed Hodtani}}\email{ghodtani@gmail.com}

% Affiliations
\affil[1]{\orgdiv{Electrical and Computer Engineering}, \orgname{Baylor University}, \orgaddress{\street{1311 S 5th St}, \city{Waco}, \postcode{76706}, \state{Texas}, \country{United States}}}

\affil*[2]{\orgdiv{Electrical Engineering}, \orgname{Ferdowsi University of Mashhad}, \orgaddress{\street{Azadi Square}, \city{Mashhad}, \postcode{9177948974}, \state{Khorasan Razavi}, \country{Iran}}}


\abstract{From both practical and theoretical viewpoints, performance analysis of communication systems using information-theoretic results is very important. In this study, first, we obtain a general achievable rate for a two-user wireless multiple access channel (MAC) with specially correlated sources as a more general version for continuous alphabet MACs, by extending the known discrete alphabet results to the wireless continuous alphabet version. Next, the impact of wireless channel coefficients correlation on the performance metrics using Copula theory, as the most convenient way for describing the dependence between several variables, is investigated. By applying the Farlie-Gumbel-Morgenstern (FGM) Copula function, we obtain closed-form expressions for the outage probability (OP) under positive/negative dependence conditions. It is shown that the fading correlation improves the OP for a negative dependence structure. Specifically, whenever the dependence structure tends to negative values, the OP decreases and the efficiency of the channel increases. Finally, the efficiency of the analytical results is illustrated numerically.}




\maketitle

\section{Introduction}\label{sec1}

Multiple access channel (MAC), first introduced by Shannon in  \cite{shannon_two-way_nodate}, has been studied widely, and the related results for the corresponding discrete memoryless (DM) channels have been extended to continuous alphabet Gaussian and wireless versions, as done for point-to-point channels. Costa \cite{costa_writing_1983} extended the Gel’fand-Pinsker theorem \cite{m_pinsker_coding_nodate} from discrete alphabet channels with transmitter side information to the Gaussian channel and proved that the capacity in the presence of the interference known at the transmitter is the same as the case with no interference. Slepian-Wolf established the capacity region of MAC with a common message for DM channels in \cite{slepian_coding_1973}. The capacity regions for GMAC with a common message were derived in \cite{prelov_asymptotic_1991}.

In wireless communication theory, the consideration of dependence structures associated with wireless channels is often neglected for the sake of simplicity \cite{biglieri_impact_2016}. This is particularly evident in the case of multiuser channels, where the physical proximity of transmitters leads to non-independent channel coefficients observed by each user. However, ignoring source correlation can lead to unrealistic models and inaccurate performance evaluations. By considering specially correlated sources, we can develop more accurate and realistic models that reflect the characteristics of real-world wireless communication systems.

For example, in multiple antennas point-to-point channels, the independence of coefficients is compromised due to physical limitations on antenna spacing. Similarly, in both single antenna and multi-antenna multi-user channels, the presence of a shared physical environment further leads to interdependence among the coefficients. Hence, it becomes crucial to thoroughly investigate the impact of the channel coefficients correlation on communication performance \cite{akuon_optimal_2016,liu_asymptotic_2010}.

The study of MAC with specially correlated sources aims to address several important problems and challenges. One primary problem is to analyze and understand the performance of wireless MAC with specially correlated sources. However, the examination of dependence structures in MAC with specially correlated sources is a challenging task that has frequently been overlooked in favor of simplicity, as mentioned earlier. Nonetheless, recognizing and addressing this aspect is vital to obtain a more accurate understanding of the channel behavior and to develop more effective approaches for maximizing performance.

One plausible approach to combine arbitrary dependence structures is the use of Copula theory \cite{roger_b_nelsen_introduction_2006}. Copulas are widely used in statistics, survival analysis, image processing, and machine learning. Recently, they have become popular in the context of wireless communication systems analysis, supported by empirical evidence \cite{peters_communications_2014}. Furthermore, the potential for designing and intelligently controlling dependence structures has served as a motivation to enhance system performance \cite{besser_copula-based_2021}. In their study \cite{noauthor_performance_nodate}, the authors explore a novel approach to enhance communication efficiency in wireless networks by utilizing Non-Orthogonal Multiple Access (NOMA) technique with amplify-and-forward (AF) relays. The proposed system operates over Nakagami-m fading channels, enabling simultaneous communication between the base station and multiple mobile users. Through comprehensive analysis, the authors investigate the system's outage behavior, providing both precise closed-form expressions and simplified bounds for the outage probability.

In \cite{tavsanoglu_wireless_2021}, the authors extensively explore the limitations and requirements related to a Hyperloop system's communication. They propose a comprehensive wireless communication system that combines two networks utilizing the capabilities of 802.11 and 5G technologies. Reference \cite{he_achieving_2020} conducts a study on secure transmission in multiple access wiretap channels where the legitimate users don't have knowledge of the instantaneous channel state information (CSI) of the wiretap channels. The authors propose a two-phase jamming-assisted secure transmission scheme that utilizes the randomness and independence of wireless channels.

The advent of Fifth generation (5G) and beyond 5G (B5G) networks necessitates connectivity for a large number of devices with high speed and low latency. Non-orthogonal multiple access (NOMA) has gained prominence in recent research as it is considered an enabling technology for 5G and B5G networks. NOMA enables multiple users to share the same time and frequency resources, offering extensive connectivity and high spectral efficiency. Reference \cite{magableh_performance_2022} examines the performance of NOMA systems assuming independent but not necessarily identical N-Nakagami-m multipath fading channels.

\section{Related Works}\label{sec2}
In a wireless point-to-point fading channel, characterized by a single transmitter and receiver, both the channel coefficient and the resulting signal-to-noise ratio (SNR) are subject to randomness. The probability distributions of these variables have been extensively investigated in the existing literature \cite{tse_fundamentals_2005,atapattu_mixture_2011}. Extensive research has been conducted to understand the statistical properties and probability distributions of channel coefficients in fading point-to-point channels with multiple antennas (MIMO) and multi-user channels like multiple access channels (MAC). These coefficients, which represent the fading characteristics of individual channel paths, introduce complexity to the system and significantly impact its overall performance. Reference \cite{s_xu_performance_nodate} explores the performance analysis of correlated multiuser multiple-input multiple-output (MIMO) systems using Copula theory. It investigates the impact of correlation among channel coefficients on system performance metrics such as capacity, outage probability, and error rates. Reference \cite{jia_reliability_2014} focuses on the outage performance analysis of correlated multiuser MIMO systems.

The study described in \cite{moeen_taghavi_impact_2016} explores the impact of relay nodes on the coverage region of MACs. The authors specifically investigate the role of relays, which act as intermediate nodes to improve signal transmission between transmitters and receivers. In their work \cite{ghadi_copula-based_2021}, the authors delve into examining the influence of fading correlation on the performance of the doubly dirty fading multiple access channel (MAC) when non-causally known side information is available at transmitters. Evaluating the coverage region is a crucial aspect when analyzing the performance of wireless communication systems.

In \cite{ghadi_role_2021}, the effect of (perfect and nonperfect) SI on the maximum coverage region of GMARC for a fixed rate of the transmitters and a fixed distance between relay and destination is analyzed, using the innovative Copula theory approach. Additionally, Reference \cite{etminan_effects_2021} presents closed-form expressions derived by the authors for average secrecy capacity (ASC), secrecy outage probability (SOP), and secrecy coverage region (SCR). Reference \cite{men_performance_2017} considers a NOMA-based relaying network over Nakagami-m fading channels, and \cite{zhang_energy-efficient_2017} studies the benefit of NOMA in enhancing energy efficiency (EE) for a multiuser downlink transmission.

\subsection{Our Contributions}\label{subsec2}
In a MAC, the inputs might be independent, specially correlated, arbitrarily correlated sources. In some applied MACs, such as wireless sensor networks, the sensed parameters, or in vehicular communications and in many uplink transmissions, the transmitted signals, are random variables, and hence, here, we study the wireless MAC with specially correlated sources. This approach is discussed in detail in \cite{bahmani_capacity_2013}. By considering the impact of channel coefficient correlation, our analysis aims to provide a deeper understanding of communication performance in wireless MAC with specially correlated sources. By quantifying the effects of dependence on system-level metrics such as OP, we can derive valuable insights to inform the design and optimization of wireless MAC protocols, resource allocation strategies, and interference management techniques.

We consider a two transmitter MAC with specially correlated sources and extend the known Slepian-Wolf (SW) discrete alphabet capacity region to the wireless version with correlated channel coefficients. Additionally, we examine the impact of correlation between wireless Rayleigh coefficients on the outage probability (OP), a critical performance metric in wireless communication. To analyze this, we employ the FGM copula function, which enables us to investigate the relationship between correlation and OP. Based on our theoretical findings, we demonstrate that when the correlated coefficients exhibit a negative structure, the OP is significantly improved compared to the scenario where the coefficients are independent.
In Section 2, we provide a comprehensive review of the MAC and Copula theory. Section 3 outlines our proposed approach and presents a detailed explanation of our work. To validate our findings, Section 4 presents numerical simulations and their results. Finally, in Section 5, we summarize our conclusions and contributions in the paper's conclusion.

\section{General Wireless Two-User MAC: An Introduction and Copula Theory Overview} \label{sec4}
In this section, we define general wireless two user MAC and review briefly the copula theory.

\subsection{The Wireless General MAC}\label{subsec4}
In the two-user wireless MAC (as illustrated in Figure~\ref{fig:Multiple Access Channel with Rayleigh Fading Coefficients}), the transmitters send the inputs $X_1$ and $X_2$, which can take various forms: independent, specially correlated, or arbitrarily correlated random variables. The received signal at the receiver, denoted as $Y$, can be defined as follows:

\begin{equation}
Y= h_1 X_1+h_2 X_2+Z
\label{eq:received_signal}
\end{equation}

Here $Z$ represents the Additive White Gaussian Noise (AWGN) with zero mean and variance $N$ which is denoted as $Z \sim \mathcal{N}(0,N)$ at the receiver. The channel coefficients $h_1$ and $h_2$ are random variables. In a specific scenario with Rayleigh fading, the squared magnitudes of $h_1$ and $h_2$ are represented as $g_1 = |h_1|^2$ and $g_2 = |h_2|^2$, respectively. These squared magnitudes follow an exponential distribution with marginal probability density functions $f(g_1) = \frac{1}{2\sigma_1^2} e^{-\frac{g_1}{2\sigma_1^2}}$ and $f(g_2) = \frac{1}{2\sigma_2^2} e^{-\frac{g_2}{2\sigma_2^2}}$, and the corresponding cumulative distribution functions are $F(g_1) = 1 - e^{-\frac{g_1}{2\sigma_1^2}}$ and $F(g_2) = 1 - e^{-\frac{g_2}{2\sigma_2^2}}$.

% Figure 1
\begin{figure}[h!]
  \centering
    \includegraphics[width= \linewidth]{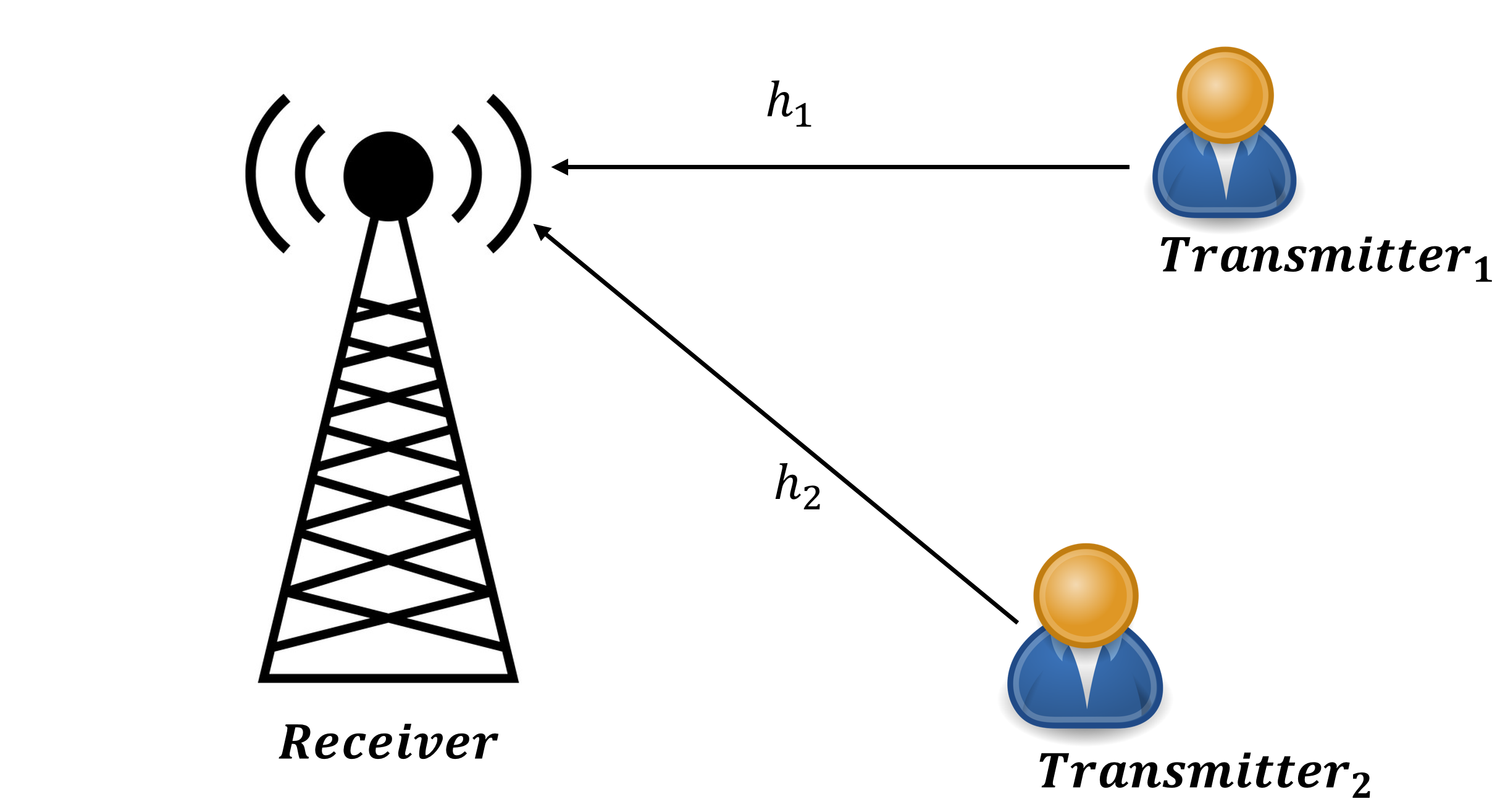}
    \caption{Multiple Access Channel with Rayleigh Fading Coefficients.}
    \label{fig:Multiple Access Channel with Rayleigh Fading Coefficients}
  \hfill
\end{figure}

\subsection{Copula Theory}\label{sub5}
In this part, we review briefly some basic definitions and properties of two-dimensional copulas.

\begin{definition}[Copula Function]\label{def1}
Let $\mathbf{V} = (V_1, V_2)$ be a vector of two random variables (RVs) with marginal cumulative distribution functions (CDFs) $F(v_j) = \Pr(V_j \leq v_j)$ for $j = 1, 2$. The relevant bivariate CDF is defined as:

\begin{equation}
F(v_1, v_2) = \Pr(V_1 \leq v_1, V_2 \leq v_2)
\end{equation}

Then, the Copula function $C(u_1, u_2)$ of $\mathbf{V} = (V_1, V_2)$ is defined on the unit hypercube $[0, 1]^2$ with uniformly distributed RVs $U_j := F(v_j)$ for $j = 1, 2$ over $[0, 1]$ as follows:

\begin{equation}
C(u_1, u_2) = \Pr(U_1 \leq u_1, U_2 \leq u_2)
\end{equation}
\end{definition}

In this section, we explore a fundamental and important theory in copula theory known as Sklar's theorem. Sklar's theorem offers a deep understanding of the role that copulas play in establishing the relationship between multivariate distribution functions and their respective univariate marginals.

\begin{theorem}[Sklar's theorem]\label{thm1}
Let $F_{V_1, V_2}(v_1, v_2)$ be a joint cumulative distribution function of random variables $V_1, V_2$ with respective marginals $F_{V_1}(v_1)$ and $F_{V_2}(v_2)$. Then there exists a copula function $C$ that satisfies the following equation:

\begin{equation}
F_{V_1, V_2}(v_1, v_2) = C(F_{V_1}(v_1), F_{V_2}(v_2))
\end{equation}
\end{theorem}

\textbf{corollary 1: }
By utilizing Sklar's theorem, the joint probability density function (PDF) $f_{V_1, V_2}(v_1, v_2)$ corresponding to the marginal cumulative distribution functions (CDFs) $F_{V_1}(v_1)$ and $F_{V_2}(v_2)$ can be expressed as follows:
\begin{equation}
f_{V_1, V_2}(v_1, v_2) = f_{V_1}(v_1) f_{V_2}(v_2) c(F_{V_1}(v_1), F_{V_2}(v_2))
\end{equation}
In equation (5), $f_{V_1}(v_1)$ and $f_{V_2}(v_2)$ represent the marginal probability density functions of $V_1$ and $V_2$ respectively. The term $c(F_{V_1}(v_1),F_{V_2}(v_2))$ corresponds to the Copula density function, which is defined as $c(F_{V_1}(v_1),F_{V_2}(v_2)) = \frac{\partial^2 C(F_{V_1}(v_1),F_{V_2}(v_2))}{\partial F_{V_1}(v_1)\partial F_{V_2}(v_2)}$.

In copula theory, various copula functions can be employed to model the dependence structure, such as Clayton, Frank, and Gumbel copulas. In this paper, we focus on the FGM (Farlie-Gumbel-Morgenstern) copula function to analyze our channel operating point. The FGM copula offers a simple case for calculating the joint PDFs, making it more practical from a computational standpoint.

\begin{definition}[FGM Copula]\label{def2}
The bivariate FGM Copula with dependence parameter $\theta_F \in [-1,1]$ is defined as:

\begin{equation}
C_F(u_1,u_2) = u_1 u_2(1 + \theta_F (1-u_1)(1-u_2))
\end{equation}
Where $\theta_F \in [-1,0)$ and $\theta_F \in (0,1]$ denote the negative and positive dependence structures respectively, while $\theta_F=0$ demonstrates the independence structure.  
\end{definition}

\section{Paper Novelties}\label{sec5}
In this section, we first describe the Gaussian SW-MAC. Then, we obtain a general achievable rate region and extend it to the Gaussian and wireless versions. Finally, after defining the OP, we derive closed-form expressions for it.

\subsection{Channel Model and Achievable Rate Region for a Two-User Gaussian SW-MAC}\label{subsec7}
We are considering a two-user MAC with specially correlated sources, as shown in Figure~\ref{fig:Multiple Access Channel with Rayleigh Fading Coefficients}. In this scenario, the sources $t_1$ and $t_2$ are sending messages $X_1$ and $X_2$, respectively. However, there are constraints on the inputs, specifically, $\mathbb{E}[|X_1|^2] \leq P_1$ and$\mathbb{E}[|X_2|^2] \leq P_2$, where $P_1$ and $P_2$ represent the maximum power limits of transmitters $t_1$ and $t_2$, respectively. As a result, the received signal $Y$ at the receiver (base station) $r$ can be expressed as follows:

\[
Y = X_1 + X_2 + Z \tag{7}
\]
where $X_1$ and $X_2$ are specially correlated and $Z$ is independent and identically distributed (i.i.d.) additive white Gaussian noise (AWGN) with zero mean and variance $N$ at the receiver.
To better understand the achievable rate region for a two-user Gaussian SW-MAC, we will now delve into the investigation with the support of theorem~{\upshape\ref{thm2}}.

\begin{theorem}[Two-User DM-MAC with Specially Correlated Sources]\label{thm2}
The inner bound rate region for two-user Discrete Memory less Multiple Access Channel (DM-MAC) with specially correlated sources is determined as follows:
\begin{subequations}
\begin{align}
R_1 &\leq I(X_1; Y \vert X_2, V_0) \tag{8a} \\
R_2 &\leq I(X_2; Y \vert X_1, V_0) \tag{8b} \\
R_1 + R_2 &\leq I(X_1, X_2; Y \vert V_0) \tag{8c} \\
R_0 + R_1 + R_2 &\leq I(X_1, X_2; Y) \tag{8d}
\end{align}
\end{subequations}
where $R_0$ represents the rate of the common message transmitted by the sources, $R_1$ denotes the rate of source 1 private message, and $R_2$ corresponds to the rate of source 2 private message. $V_0$ represents the common message sent by the sources.    
\end{theorem}

\begin{proof}[Proof of Theorem~{\upshape\ref{thm2}}]
The proof of Theorem ~{\upshape\ref{thm2}}, which establishes the validity and correctness of the theorem, is presented in reference \cite{sanyal_game_2021}.
\end{proof}

% Theorem 3:
\begin{theorem}[Two-User Gaussian MAC with Specially Correlated Sources]\label{thm3}
The inner bound rate region for two-user Gaussian MAC with specially correlated sources is obtained as follows:
\begin{subequations}
\begin{align}
R_1 &\leq \frac{1}{2} \log \left(1 + \frac{P_1 - P_0}{N}\right) \tag{9a} \\
R_2 &\leq \frac{1}{2} \log \left(1 + \frac{P_2 - P_0}{N}\right) \tag{9b} \\
R_1 + R_2 &\leq \frac{1}{2} \log \left(1 + \frac{P_1 + P_2 - 2P_0}{N}\right) \tag{9c} \\
R_0 + R_1 + R_2 &\leq \frac{1}{2} \log \left(1 + \frac{P_1 + P_2 + 2P_0}{N}\right) \tag{9d}
\end{align}
\end{subequations}   
\end{theorem}

\textbf{Explanatory note for the proof of Theorem~{\upshape\ref{thm3}}}:
theorem ~{\upshape\ref{thm3}} extends the findings of Theorem ~{\upshape\ref{thm2}} to the continuous alphabet Gaussian channel. As depicted in Figure \ref{fig:Multiple Access Channel with Rayleigh Fading Coefficients}, the additive white Gaussian noise (AWGN) component, denoted as $Z$, exhibits zero mean and a variance of $N$. It represents the inherent noise present in the communication system, the individual allocated powers for the sources are denoted as $P_1$, $P_2$. The transmitted signals can be written as:

\begin{align}
X_1 &= V_0 + a_1 W_1 \tag{10} \\
X_2 &= V_0 + a_2 W_2 \tag{11}
\end{align}

where $W_1$ and $W_2$ represent the private messages sent by sources 1 and 2, respectively. These messages are assumed to be independent of $V_0$ with zero mean. Furthermore, we assume that $E[V_0^2] = P_0$, $E[W_1^2] = P_{11}$, and $E[W_2^2] = P_{21}$. Then, to determine the parameters $a_1$ and $a_2$, we calculate them as follows:

\begin{align}
a_1 &= \sqrt{\frac{P_1 - P_0}{P_{11}}} \tag{12} \\
a_2 &= \sqrt{\frac{P_2 - P_0}{P_{21}}} \tag{13}
\end{align}

By substituting the given parameter values, the transmitted signals can be further simplified as follows:
$X_1 = V_0 + \sqrt{\frac{P_1 - P_0}{P_{11}}} W_1$ and $X_2 = V_0 + \sqrt{\frac{P_2 - P_0}{P_{21}}} W_2$. Now we continue the proof.

\begin{proof}[Proof of Theorem~{\upshape\ref{thm3}}]
The proof of Theorem ~{\upshape\ref{thm3}} is given as follows:

The first transmitter rate:
\begin{align}
R_1 &\leq I(X_1; Y \vert X_2, V_0) \notag \\
    &= h(Y \vert X_2, V_0) - h(Y \vert X_1, X_2, V_0) \notag \\
    &= h(X_1 + X_2 + Z \vert X_2, V_0) - h(X_1 + X_2 + Z \vert X_1, X_2, V_0) \notag \\
    &= h\left(V_0 + \sqrt{\frac{P_1 - P_0}{P_{11}}} W_1 + V_0 + \sqrt{\frac{P_2 - P_0}{P_{21}}} W_2 + Z \bigg| X_2, V_0\right) \notag \\
    &\quad - h\left(V_0 + \sqrt{\frac{P_1 - P_0}{P_{11}}} W_1 + V_0 + \sqrt{\frac{P_2 - P_0}{P_{21}}} W_2 + Z \bigg| X_2, X_1, V_0\right) \notag \\
    &\stackrel{i}{=} h\left(\sqrt{\frac{P_1 - P_0}{P_{11}}} W_1 + Z \bigg| X_2, V_0\right) - h(Z \bigg| X_2, X_1, V_0) \notag \\
    &\stackrel{ii}{=} h\left(\sqrt{\frac{P_1 - P_0}{P_{11}}} W_1 + Z\right) - h(Z) \notag \\
    &= \frac{1}{2} \log \left(2\pi e \left(\sigma_{\sqrt{\frac{P_1 - P_0}{P_{11}}} W_1 + Z}^2\right)\right) - \frac{1}{2} \log \left(2\pi e \sigma_z^2\right) \notag \\
    &= \frac{1}{2} \log \left(1 + \frac{P_1 - P_0}{N}\right) \label{eq:R1_final} \tag{14}
\end{align}
$(i)$ The term due to knowing $(X_1, X_2, V_0)$ implies that the knowledge of these variables does not influence the entropy.
$(ii)$ The independence of $(Z, W_1)$ of $(X_2, V_0)$ and the independence of $(Z)$ of $(X_1, X_2, V_0)$ allow us to separate these variables when calculating the entropy.

The second transmitter rate:
\begin{align}
R_2 &\leq I(X_2; Y \vert X_1, V_0) \notag \\
    &= h(Y \vert X_1, V_0) - h(Y \vert X_1, X_2, V_0) \notag \\
    &= h(X_1 + X_2 + Z \vert X_1, V_0) - h(X_1 + X_2 + Z \vert X_1, X_2, V_0) \notag \\
    &= h\left(V_0 + \sqrt{\frac{P_1 - P_0}{P_{11}}} W_1 + V_0 + \sqrt{\frac{P_2 - P_0}{P_{21}}} W_2 + Z \bigg| X_1, V_0\right) \notag \\
    &\quad - h\left(V_0 + \sqrt{\frac{P_1 - P_0}{P_{11}}} W_1 + V_0 + \sqrt{\frac{P_2 - P_0}{P_{21}}} W_2 + Z \bigg| X_1, X_2, V_0\right) \notag \\
    &\stackrel{i}{=} h\left(\sqrt{\frac{P_2 - P_0}{P_{21}}} W_2 + Z \bigg| X_1, V_0\right) - h(Z \bigg| X_1, X_2, V_0) \notag \\
    &\stackrel{ii}{=} h\left(\sqrt{\frac{P_2 - P_0}{P_{21}}} W_2 + Z\right) - h(Z) \notag \\
    &= \frac{1}{2} \log \left(2\pi e \left(\sigma_{\sqrt{\frac{P_2 - P_0}{P_{21}}} W_2 + Z}^2\right)\right) - \frac{1}{2} \log \left(2\pi e \sigma_z^2\right) \notag \\
    &= \frac{1}{2} \log \left(1 + \frac{P_2 - P_0}{N}\right) \tag{15}
\end{align}
The reason for $(i)$ is that knowing $X_1,X_2,V_0$ does not affect the entropy since these variables are already known. The reason for $(ii)$ is the independence of $(Z,W_2)$ from $X_1,V_0$ and the independence of $(Z)$ of $(X_1,X_2,V_0)$.

The sum of the private message rates of two transmitters:

\begingroup
\scriptsize % or \tiny for the smallest font size
\begin{align}
R_1 + R_2 &\leq I(X_1, X_2; Y \vert V_0) \notag \\
         &= h(Y \vert V_0) - h(Y \vert V_0, X_1, X_2) \notag \\
         &= h(X_1 + X_2 + Z \vert V_0) - h(X_1 + X_2 + Z \vert X_1, X_2, V_0) \notag \\
         &= h\left(\sqrt{\frac{P_1 - P_0}{P_{11}}} W_1 + V_0 + \sqrt{\frac{P_2 - P_0}{P_{21}}} W_2 + V_0 + Z \bigg| V_0\right) \notag \\
         &\quad - h\left(\sqrt{\frac{P_1 - P_0}{P_{11}}} W_1 + V_0 + \sqrt{\frac{P_2 - P_0}{P_{21}}} W_2 + V_0 + Z \bigg| X_1, X_2, V_0\right) \notag \\
         &= {}^{\text{i}} h\left(\sqrt{\frac{P_1 - P_0}{P_{11}}} W_1 + \sqrt{\frac{P_2 - P_0}{P_{21}}} W_2 + Z \bigg| V_0\right) - h(Z \bigg| X_1, X_2, V_0) \notag \\
         &= {}^{\text{ii}} h\left(\sqrt{\frac{P_1 - P_0}{P_{11}}} W_1 + \sqrt{\frac{P_2 - P_0}{P_{21}}} W_2 + Z\right) - h(Z) \notag \\
         &= \frac{1}{2} \log \left(2\pi e \left(\sigma^2_{\sqrt{\frac{P_1 - P_0}{P_{11}}} W_1 + \sqrt{\frac{P_2 - P_0}{P_{21}}} W_2 + Z}\right)\right) - \frac{1}{2} \log \left(2\pi e N\right) \notag \\
         &= \frac{1}{2} \log \left(1 + \frac{P_1 + P_2 - 2P_0}{N}\right) \tag{16}
\end{align}
\endgroup

The reasons for steps (i) and (ii) are:

\begin{enumerate}
    \item The entropy is not influenced by knowing $\mathbf{X}_1, \mathbf{X}_2$, and $V_0$ since these variables are already known.
    \item The independence of $(Z, \mathbf{W}_1)$ from $(\mathbf{X}_2, V_0)$ and the independence of $Z$ from $(\mathbf{X}_1, \mathbf{X}_2, V_0)$.
\end{enumerate}

The sum of the common and private Rates:

\begin{multline}
    R_0+R_1+R_2 \leq I(X_1,X_2;Y) = h(Y) - h(Y|X_1,X_2) \\
    = h(X_1+X_2+Z) - h(X_1+X_2+Z|X_1,X_2) \\
    = h\left(\sqrt{\frac{P_1-P_0}{P_{11}}} W_1+V_0+\sqrt{\frac{P_2-P_0}{P_{21}}} W_2+V_0+Z\right) - h(Z) \\
    = \frac{1}{2} \log\left(2\pi e \left(\sigma_{\sqrt{\frac{P_1-P_0}{P_{11}}} W_1+V_0+\sqrt{\frac{P_2-P_0}{P_{21}}} W_2+V_0+Z}\right)^2\right) \\
    - \frac{1}{2} \log(2\pi e (\sigma_Z)^2) \\
    = \frac{1}{2} \log\left(1+\frac{P_1+P_2+2P_0}{N}\right)
    \tag{17}
\end{multline}

\end{proof}

\subsection{Wireless SW-MAC}\label{subsec8}

In a two-user Rayleigh fading MAC with specially correlated sources as shown in Figure~\ref{fig:Multiple Access Channel with Rayleigh Fading Coefficients}, transmitters (users) $t1$ and $t2$ transmit their respective messages $X1$ and $X2$. The inputs are restricted such that $E[|X1|]^2 < P1$ and $E[|X2|]^2 < P2$, where $P1$ and $P2$ are the maximum power of transmitters $t1$ and $t2$ respectively. Accordingly, the received signal $Y$ at the receiver, denoted as $r$ (base station), can be expressed as follows:

\begin{equation}
    Y = h_1 X_1 + h_2 X_2 + Z \tag{18}
\end{equation}

At the receiver, the signal $Y$ is impacted by independent identically distributed (i.i.d.) Additive White Gaussian Noise (AWGN) denoted as $Z$, with a mean of zero and variance $N$. Additionally, the fading channel coefficients $h_1$ and $h_2$ correspond to Rayleigh fading processes and are assumed to be correlated.

\subsection{The Capacity Inner Bound for Wireless Multiple Access Channels with Specially Correlated Sources
Referencing (18), Theorem 3 is extended to the wireless version as follows}
The capacity region was determined through the following calculation process:

\begin{align}
    &R_1 \leq \frac{1}{2} \log(1 + \frac{|h_1|^2 (P_1 - P_0)}{N}) \tag{19} \\
    &R_2 \leq \frac{1}{2} \log(1 + \frac{|h_2|^2 (P_2 - P_0)}{N}) \tag {20} \\
    &R_1 + R_2 \leq \frac{1}{2} \log(1 + \frac{|h_1|^2 (P_1 - P_0) + |h_2|^2 (P_2 - P_0)}{N}) \tag {21} \\
    &R_0 + R_1 + R_2 \leq \frac{1}{2} \log(1 + \frac{|h_1|^2 P_1 + |h_2|^2 P_2 + 2|h_1||h_2|P_0}{N}) \tag {22}
\end{align}

\textbf{Corollary 2:} As seen intuitively, the above relations (19)-(22) with $P_0 = 0$ are reduced to the corresponding region for wireless MAC with independent sources.

\subsection{Definition and Analysis of Outage Probability (OP)}\label{subsec9}
In this section, we focus on two key aspects. First, we define the OP as a fundamental metric for analyzing wireless channel. Second, we employ the FGM copula to calculate the OP, enabling us to quantify the communication performances. 

\begin{definition}[OP]\label{def3} OP plays a crucial role in assisting the performance of the wireless communication systems, and is defined as the probability that the channel capacity is less than a certain information rate $R_0 > C$. Mathematically, OP can be defined as follows:
\begin{equation}
 P[R_0 > C] = P\left[\frac{|h|^2}{N} < \frac{1}{P} \left(2^{nR_0} - 1\right)\right] \tag{23}
\end{equation}
Where $R_0$ is the intended rate, and $C$ is the channel capacity. 
\end{definition}
    
\section{Outage Probability (OP) for the Wireless SW-MAC}\label{sec}
The OP for specially correlated sources MAC can be expressed using the following theorem:

% Theorem4:
\begin{theorem}[The OP in wireless specially correlated sources MAC is given ] \label{thm4}
\begin{align}
    P_{\text{out}} = 1 - \Bigg[ & \frac{\lambda_2 e^{-\left(\frac{\lambda_1 \gamma}{P_1 - P_0}\right)}}{(\lambda_2 - \lambda_1 P)} 
    + \theta \left( \frac{\lambda_2 e^{-\left(\frac{\gamma \lambda_1}{(P_1 - P_0)}\right)}}{(\lambda_2 - \lambda_1 P)} 
    - \frac{2\lambda_2 e^{-\left(\frac{\gamma \lambda_1}{(P_1 - P_0)}\right)}}{(2\lambda_2 - P \lambda_1)} \right. \nonumber \\
    & \left. - \frac{\lambda_2 e^{-\left(\frac{2 \gamma \lambda_1}{(P_1 - P_0)}\right)}}{(\lambda_2 - 2P \lambda_1)} 
    + \frac{\lambda_2 e^{-\left(\frac{2 \gamma \lambda_1}{(P_1 - P_0)}\right)}}{(\lambda_2 - \lambda_1 P)} \right) \Bigg]
    \tag{24}
\end{align}
Where $\lambda_1 = \frac{1}{2\sigma_1^2}$, $\lambda_2 = \frac{1}{2\sigma_2^2}$, and $P = \frac{P_2 - P_0}{P_1 - P_0}$.
\end{theorem}

\begin{proof}[Proof of Theorem~{\upshape\ref{thm4}}]
    
First, as claimed by the OP definition, we have the following:

\begin{align}
P_\text{out} &= P[R \geq R_1 + R_2] \notag \\
             &= P\left[R \geq \frac{1}{2} \log\left(1 + \frac{|h_1|^2 (P_1 - P_0) + |h_2|^2 (P_2 - P_0)}{N}\right)\right] \notag \\
             &= P\left[|h_1|^2 (P_1 - P_0) + |h_2|^2 (P_2 - P_0) \leq N(2^{2R} - 1)\right] \notag \\
             &= P[C(P_1 - P_0) + D(P_2 - P_0) \leq \gamma] \notag \\
             &= 1 - \int_0^\infty \int_{(\gamma - d(P_2 - P_0))/(P_1 - P_0)}^\infty f_{C,D}(c,d) \, dc \, dd \tag{25}
\end{align}

Where $D = |h_1|^2$, $C = |h_2|^2$, and $\gamma = N(2^{2R} - 1)$.

Since the coefficients of the fading channels are assumed to be correlated, we need the joint PDF of \(|h_1|^2\) and \(|h_2|^2\) to compute the OP. According to earlier assumptions, \(h_1\) and \(h_2\) are correlated with a Rayleigh distribution, and the square of a Rayleigh distribution (i.e., \(|h_1|^2, |h_2|^2\)) is an exponential distribution. Therefore, we have the density distributions for \(|h_1|^2\) and \(|h_2|^2\) respectively, as \(f_C(c) = \frac{1}{2\sigma_1^2} e^{-c / (2\sigma_1^2)}\) and \(f_D(d) = \frac{1}{2\sigma_2^2} e^{-d / (2\sigma_2^2)}\) by parameters \(\sigma_1^2\) and \(\sigma_2^2\). Now by inserting them into (5), and using (6), we can show the FGM copula as follows:

\begingroup
\scriptsize % or \tiny for an even smaller font size
\begin{align*}
C(F_C(c), F_D(d)) &= F_C(c) F_D(d) [1 + \theta (1 - F_C(c))(1 - F_D(d))] \\
&= (1 - e^{-c/(2\sigma_1^2)}) (1 - e^{-d/(2\sigma_2^2)}) \left[1 + \theta e^{-c/(2\sigma_1^2)} e^{-d/(2\sigma_2^2)}\right] \\
&= 1 - e^{-d/(2\sigma_2^2)} - e^{-c/(2\sigma_1^2)} + e^{-c/(2\sigma_1^2)} e^{-d/(2\sigma_2^2)} \\
&\quad + \theta \left[e^{-c/(2\sigma_1^2)} e^{-d/(2\sigma_2^2)} - e^{-c/(2\sigma_1^2)} e^{-2d/(2\sigma_2^2)} \right. \\
&\qquad \left. - e^{-2c/(2\sigma_1^2)} e^{-d/(2\sigma_2^2)} + e^{-2c/(2\sigma_1^2)} e^{-2d/(2\sigma_2^2)}\right] \tag{26}
\end{align*}
\endgroup

By calculating \( f_{C,D}(c,d) \) using (5), we have:

\begingroup
\scriptsize % or \tiny for an even smaller font size
\begin{align}
f_{C,D}(c,d) &= \frac{1}{4\sigma_1^2 \sigma_2^2} e^{-c/(2\sigma_1^2)} e^{-d/(2\sigma_2^2)} \notag \\
&\quad + \theta \left[\frac{1}{4\sigma_1^2 \sigma_2^2} e^{-c/(2\sigma_1^2)} e^{-d/(2\sigma_2^2)} - \frac{2}{4\sigma_1^2 \sigma_2^2} e^{-c/(2\sigma_1^2)} e^{-2d/(2\sigma_2^2)} \right. \notag \\
&\qquad \left. - \frac{2}{4\sigma_1^2 \sigma_2^2} e^{-2c/(2\sigma_1^2)} e^{-d/(2\sigma_2^2)} + \frac{4}{4\sigma_1^2 \sigma_2^2} e^{-2c/(2\sigma_1^2)} e^{-2d/(2\sigma_2^2)} \right] \tag{27}
\end{align}
\endgroup

Now by inserting (27) into (25), we have:

\begin{align}
P_{\text{out}} &= 1 - \int_0^\infty \int_{\frac{\gamma - d(P_2 - P_0)}{P_1 - P_0}}^\infty \frac{1}{4\sigma_1^2 \sigma_2^2} \left[ e^{-c/(2\sigma_1^2)} e^{-d/(2\sigma_2^2)} \right. \notag \\
&\quad + \theta \left( e^{-c/(2\sigma_1^2)} e^{-d/(2\sigma_2^2)} - 2 e^{-c/(2\sigma_1^2)} e^{-2d/(2\sigma_2^2)} \right. \notag \\
&\qquad \left. \left. - 2 e^{-2c/(2\sigma_1^2)} e^{-d/(2\sigma_2^2)} + 4 e^{-2c/(2\sigma_1^2)} e^{-2d/(2\sigma_2^2)} \right) \right] \, dc \, dd \tag{28}
\end{align}

And by calculating the above integrals, the OP is given by (24).
\end{proof}

\section{Simulation Results and Comparison}
In this section, we present the numerical results of the obtained performance for our channel model. This simulation compares the performance of the channel under correlated and uncorrelated Rayleigh fading coefficients. The list of channel parameters can be found in Tables~\ref{tab1}.

% Table 1
\begin{table}[h]
\caption{Channel Parameters}\label{tab1}%
\begin{tabular}{@{}lll@{}}
\toprule
Parameters & Descriptions  & Values \\
\midrule
\( P_1 \)   &  Power of First Transmitter   & 1, 5, 10 Watts \\
\( P_2 \)   & Power of Second Transmitter  & 1, 5, 10, Watts  \\
\( \alpha \) &  Path-loss exponent    & 2.8  \\
\( R_1 \) & Transmission rate of First Transmitter & 0.44 Mbps \\  % Row 4
\( R_2 \) & Transmission rate of Second Transmitter & 1.75 Mbps \\  % Row 5
\( N \) & Power of noise & \( 10^{-5} \) \\  % Row 6
\botrule
\end{tabular}
\end{table}

\begin{figure}[ht]
  \centering
  \includegraphics[width=\linewidth]{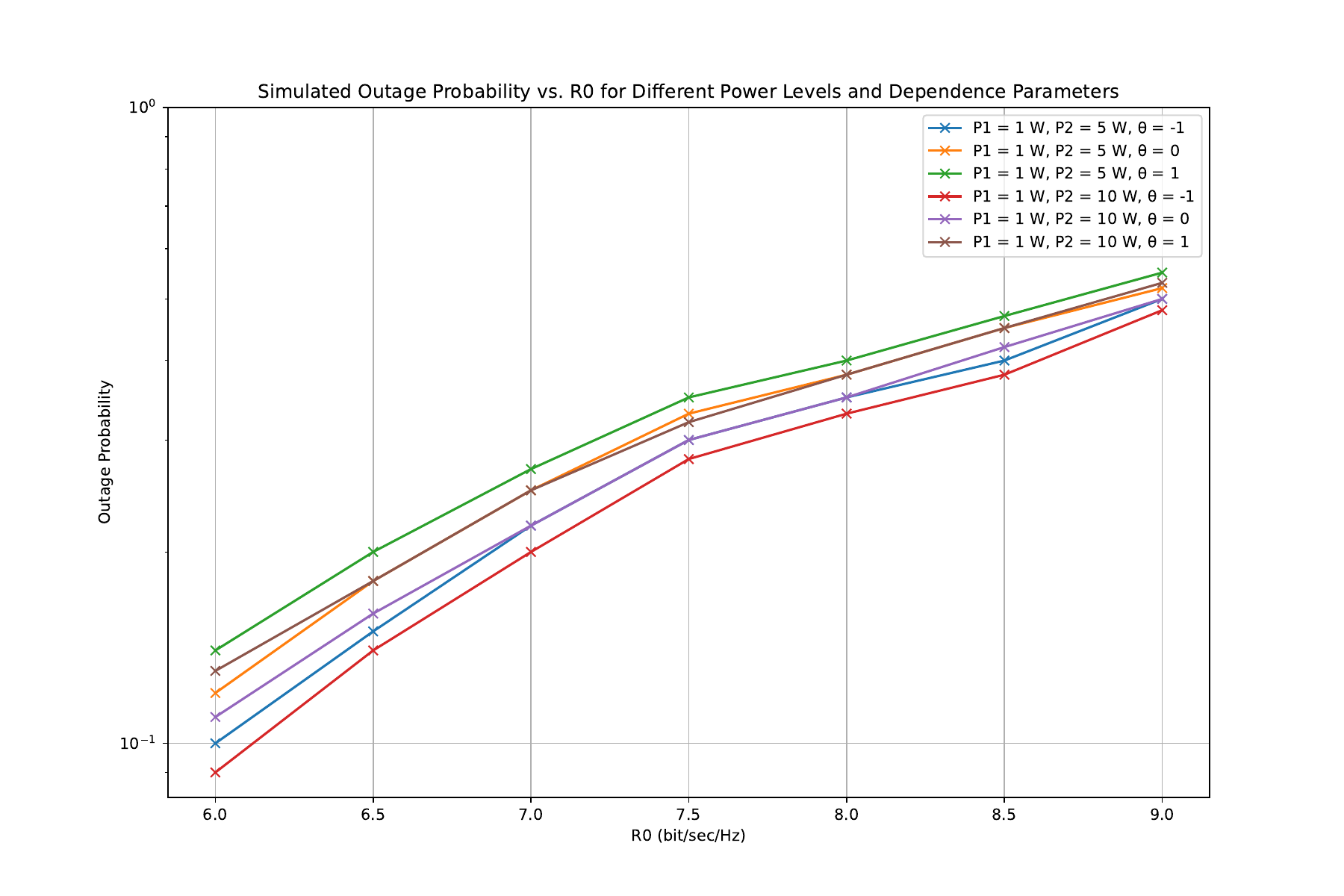}
  \caption{Outage Probability vs $R_0$ For $P2=5$ and $10$ watts, and $P1=1$ watt, and Various Dependence Parameters.}
  \label{fig:p_outage_R_p1<p2}
\end{figure}

\begin{figure}[htbp]
  \centering
  \includegraphics[width=\linewidth]{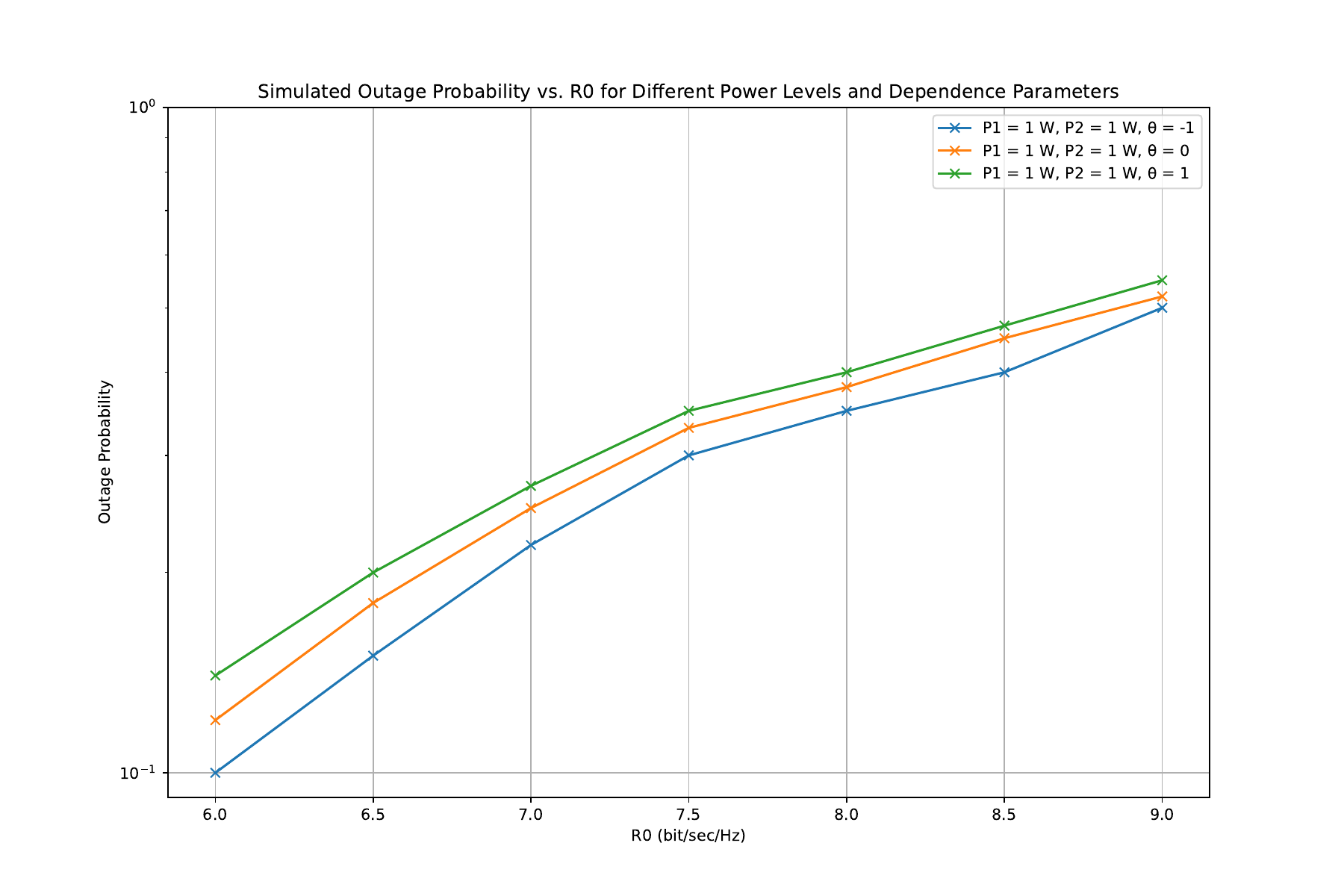}
  \caption{Outage Probability vs $R_0$ For $P1=1$ watt, and $P2=1$ watt, and Various Dependence Parameters.}
  \label{fig:p_outage_R_p1=p2}
\end{figure}

\begin{figure}[htbp]
  \centering
  \includegraphics[width=\linewidth]{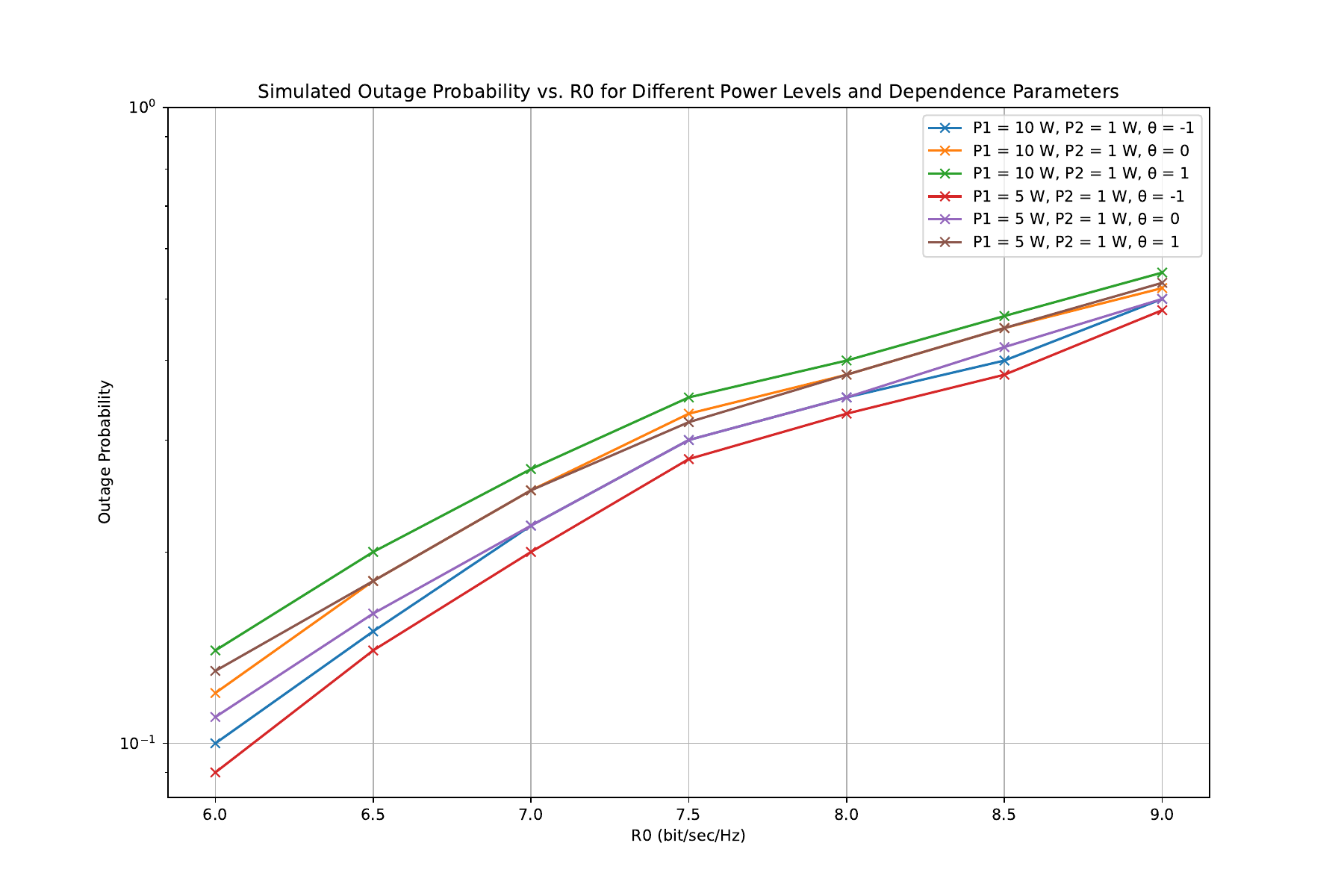}
  \caption{Outage Probability vs $R_0$ For $P1=5$ and $10$ watts, and $P2=1$ watt, and Various Dependence Parameters.}
  \label{fig:p_outage_R_p1>p2}
\end{figure}

We illustrated  the impact of the required threshold information rate ($R$) on the OP in three different realistic scenarios represented as Figure ~\ref{fig:p_outage_R_p1<p2}, Figure ~\ref{fig:p_outage_R_p1=p2}, and Figure ~\ref{fig:p_outage_R_p1>p2}. These simulations examine variations in power levels, showing how the OP changes with increasing $R$ and how it depends on the dependence parameter $\theta$. 
In our first experiment, we investigated power levels of $P2=5$ and $10$ watts, and $P1=1$ watt. We observed that for negative dependence structures ($\theta \in [-1,0)$), the performance is better compared to positive dependence structures ($\theta \in (0,1]$). Additionally, in the case of uncorrelated fading ($\theta = 0$), there is a lower OP compared to the positive dependence structure, resulting in less interference. As ($P2$) increases from $5$ W to $10$ W, the OP decreases, indicating improved performance. See  Figure ~\ref{fig:p_outage_R_p1<p2}.
Next, we consider the scenario where the power of the first transmitter equals the power of the second transmitter ($P = 1$ watt). It is evident that negative dependence structures outperform positive dependence structures, as shown in Figure ~\ref{fig:p_outage_R_p1=p2}. 
When the power of the first transmitter is greater than that of the second transmitter, as shown in  Figure ~\ref{fig:p_outage_R_p1>p2}, the channel exhibits better performance in terms of OP under negative dependence structures ($\theta \in [-1,0]$) compared to positive dependence structures ($\theta \in (0,1]$). Conversely, independent structures ($\theta = 0$) favor the uncorrelated fading  case, illustrating  superior performance compared to positive dependence structures. Most notably, as the power ($P1$) decreases from 10 W to 5 W, the OP decreases, indicating enhanced performance.

\FloatBarrier  
\section{Conclusion}
We studied the communication performance of wireless SW-MAC, which is a generalized version of continuous alphabet MACs. In this study, we assumed the channel coefficients to be correlated. Using the FGM copula function to account for negative and positive correlations, we developed mathematical expressions and upper/lower bounds for the OP. We analyzed the influence of coefficient correlation under both negative and positive dependence structures. The results confirm that negative correlated coefficients have a positive impact on the performance of wireless MAC compared to the scenario with independent coefficients.

\subsection*{Future Work}
The understanding of SW-MAC has a significant impact on important issues such as channel assignment, power allocation, and energy efficiency. The results of this study open up several avenues for further research, particularly in exploring correlated channel coefficients in wireless communications. This exploration can lead to more efficient designs and reduced computational intensity and memory requirements in these fields.
The findings prompt important questions about the impact of increasing channels and users on each other, which will be thoroughly investigated in future studies.
It would be beneficial for future research to concentrate on exploring Rician distribution channel coefficients and utilizing different copula functions to evaluate the performance of wireless channel modes. By delving into these areas, subsequent studies can expand on the existing foundation, potentially resulting in significant improvements in the efficiency and effectiveness of SW-MACs in complex communication environments.

%\bibliographystyle{sn-mathphys-num} % Specify the bibliography style
\bibliography{sn-bibliography}% common bib file

\end{document}